\documentclass[12pt,reqno]{amsart}

\usepackage{amsmath,amsthm,amscd,amssymb,bbm,graphicx, color}

\usepackage{mathrsfs} 
\usepackage{geometry}          
\usepackage{graphics,multicol}
\usepackage{epstopdf}
\usepackage{enumerate}

\usepackage{dsfont}
\newtheorem{remark}{Remark}

\newtheorem{theorem}{Theorem}

\newtheorem{lemma}[theorem]{Lemma}

\newtheorem{corollary}[theorem]{Corollary}
\usepackage{hyperref}

\newcommand{\p}{{\mathbb P}}		
\newcommand{\ER}{\text{ER}}				

\usepackage{natbib}

\title{Rumors on evolving graphs through stationary times}
\author{Vicenzo Bonasorte}
\thanks{This work is part of the author's Masters dissertation, written under the supervision of Prof. Sandro Gallo and supported by Coordenação de Aperfeiçoamento de Pessoal de Nível Superior – Brasil (CAPES) – Finance Code 001 and by grants 2022/08843-6 and 2025/00805-6 of the São Paulo Research Foundation (FAPESP).\\
E-mail address: $<$vicenzop@usp.br$>$.}
\begin{document}

\begin{abstract}
We study rumor spreading in dynamic random graphs. Starting with a single informed vertex, the information flows until it reaches all the vertices of the graph (completion), according to the following process. At each step $k$, the information is propagated to neighbors of the informed vertices, in the $k$-th generated random graph. The way this information propagates from vertex to vertex at each step will depend on the ``protocol". We provide a method based on strong stationary times to study the completion time when the graphs are Markovian time dependent, using known results of the literature for independent graphs. The concept of strong stationary times is then extended to non-Markovian Dynamics using coupling from the past algorithms. This allows to extend results on completion times for non-Markov dynamics.\\

\end{abstract}

\maketitle
\tableofcontents    

\section{Introduction}

This paper focuses on Randomized Rumor Spreading protocols, a series of algorithms designed to disseminate a piece of information in a network. Some of the more important versions of this algorithm are the Push protocol, Pull protocol and Push-Pull protocol. In all of them we start with a single informed node. The protocol follows in a series of synchronous time-steps (or rounds). The Push and Pull variants are symmetric. While in the Push protocol each informed node sends the information to a single neighbor chosen uniformly at random, in the Pull protocol uninformed nodes are the active ones, they try to pull the rumor from a single neighbor chosen uniformly at random. Finally, the Push-Pull protocol simply combines Push and Pull, both informed and uninformed nodes are active.

In the literature one of the problems having most interest is the so-called \emph{completion time} problem: how many time-steps of the protocol are needed until the information reaches every vertex? Note that this question only makes sense in two main configurations. Connected static networks and dynamic networks.

\vspace{0.1cm}

\paragraph{\bf{Communication protocols on static networks.}}
These protocols were first introduced to propagate rumors on static networks. The Push protocol, for example, has been considered on the complete graph (see \cite{grimmett-85}), random graphs and the hypercube (see \cite{feige-90}).The work on static networks usually relies on spectral properties of the underlying graph. For example, the conductance of the graph was considered by \cite{panconesi-10} to prove a logarithmic upper bound for the completion time of the Push-Pull protocol. \cite{alexandre} considers the vertex expansion to obtain upper bounds on the completion time of the Push protocol and the Pull protocol on regular graphs. 

Many real life networks are not static. The connections in wireless and mobile networks, in a network of users exchanging e-mails or in social groups are in constant change. Since the above theory relies on the structure of the underlying graph, the study of rumor spreading on dynamic random networks requires new methods. Dynamic networks also allow us to study a different algorithm: the Flood protocol, in which informed nodes inform every neighbor each round. For static graphs the Flood protocol either takes a deterministic amount of steps until completion, or the information never reaches all nodes (non-connected graphs).
\vspace{0.1cm}

\paragraph{\bf{Rapid overview of the literature on random dynamic networks.}}

\cite{clementi2016rumor} studied the Push version of this algorithm in an underlying dynamic network. In one setting, each graph is a newly sampled (independently of everything else) Erd\H{o}s-Rényi random graph, $ER_n(p)$. They showed that the completion time for the Push protocol is $O(\frac{\log n}{\min\{1,np\}})$ with high probability, that is with probability decreasing polynomially to 0 as $n$ diverges. \cite{doerr} improved this bound when the edge probability is $\frac{a}{n}$, where $n$ is the number of edges and $a$ is a positive constant. They introduced a method for analyzing rumor spreading that does not rely on the specific definition of the protocol, under some sufficient conditions they provide the expectation and a large deviations bound for the completion time for several settings, apart from additive constants. 

\cite{daknama} uses the method of \cite{doerr} to analyze rumor spreading in a sequence of independent Erd\H{o}s-Rényi graphs for the Pull and Push-Pull protocol, when the edge probability is of the form $\frac{a}{n}$, $a > 0$. He proves that the expected completion time for both the Pull protocol and the Push-Pull is $O(\log n)$.

\cite{clementiflood} and \cite{clementi2016rumor} analyzed rumor spreading in edge-Markovian random graphs. In this dynamic graph, there is a binary Markov chain attached to each possible edge (independent of one another). If an edge was absent at time $t$, it will be present at time $t+1$ with probability $p$, if it was present at time $t$, it will be absent at time $t+1$ with probability $q$. In the former they considered the Flood protocol and prove that for any $p$ and $q$ the completion time is $O(\frac{\log n}{\log(1+np)})$ with high probability. In the latter they considered the rumor spreading according to the Push Protocol and prove that, for $p \geq \frac{1}{n}$ and when $q$ is constant (does not depend on $n$), the completion time is $O(\log n)$, with high probability.

As far as we know, non-Markovian settings have not been considered before in the literature and constitute a natural next step. This text provides a simple strategy to bound the completion time in such cases, as well.

\paragraph{\bf{Main contributions of this text.}} Consider a sequence of random graphs in which the set of vertices is fixed, and the edges evolve, independently of each other, according to a stationary process indicating whether at each time the edge is present or absent. We provide a simple general method allowing us to study the completion time of different protocols of rumor propagation on these general dynamical random graphs using what is known for the same protocols when the edges dynamic is i.i.d. The idea of the method is to find, in the edges dynamic, a sequence of random times similar to what is called strong stationary times in the literature of Markov chains: instants at which the graphs are independent and distributed according to the marginal of the time stationary law.  If the distance between these times has a sufficiently light tail, we directly get upper bounds on the completion time of the protocol on the dependent sequence of graphs using what is known about the completion time in the independent sequence of graphs with equal marginal.
First, using the existing theory of strong uniform times we are, under some sufficient conditions on the edge transition matrix, able to bound the completion time of the Push, Pull, Flooding and Push-Pull protocols in the Markov dynamic. Bounds for the Push and Push-Pull protocols had not been done before in the literature, as far as we know. In the Push protocol setting we prove that the completion time is $O(n^{k-1}\log n)$ when $p = \frac{a}{n^k}$ and $q = 1$, for constants $a > 0$ and $k >1$, which is also new in the literature. Second, we extend the concept of strong stationary times to any non-Markovian process that can be perfectly simulated through coupling from the past algorithms. This is a problem having interest of its own and allows us to illustrate possibilities of randomized rumor spreading out of the scope of Markovianity. We give an explicit example of the case where the edge dynamics is a renewal process.

\section{A strategy based on strong stationary times}
\label{literature}

Consider a set of $n\in\mathbb N$ vertices, denoted $[n]$. Consider also the set of possible edges $E^{(n)}=\{(x,y):x,y\in[n]\}$, which has cardinal $|E^{(n)}|=C_2^n$. For any $e\in E^{(n)}$, consider a binary stochastic chain ${\bf X}(e)=(X_t(e))_{t\ge0}$, and assume that the chains ${\bf X}(e),e\in E^{(n)}$ are independent. The event $X_t(e)=1$ means ``the edge $e$ exists at time $t$''. If we define $\mathcal E^{(n)}_t:=\{e\in E^{(n)}:X_t(e)=1\},t\ge0$, we get a  sequence of random graphs $([n],\mathcal E^{(n)}_t),t\ge0$.  

For instance, if for each $e$ the chain $X_t(e),t\ge0$ is i.i.d., and if $\p(X_t(e)=1)=p$ for any $e\in[n]$, then we obtain a sequence of independent copies of the Erd\H{o}s-R\'enyi graph $\ER_n(p)$.

Each node is either informed or uninformed. Starting with a single informed node, the rumor spreads according to an information spreading protocol.

How many steps it takes for a protocol to inform all the nodes? This number of steps is called  \emph{completion time}.  

The proof strategies for completion time bounds is heavily dependent on the underlying random graph and the properties of the protocol. The next Lemma \ref{strategy} explores a new strategy to attack these problems in dependent graphs using the results in the i.i.d. case.

\label{strategy_section}
\begin{lemma}
\label{strategy}
 Suppose that we have a non-independent sequence of graphs $([n],\mathcal E^{(n)}_t),t\ge0$, but still, identically distributed, that is, for any $e\in [n]$, ${\bf X}(e)$ is stationary. Let $\pi = (\pi(0),\pi(1))$ denote the 1-step marginal stationary distribution of ${\bf X}(e)$. Consider an information spreading protocol. Suppose that its completion time in the i.i.d. case is $O(r(n, \pi(1)))$ with high probability. If we are able to extract  random sequence $t_1,t_2,\ldots$ of increasing integers such that: 
\begin{enumerate}
\item $([n],\mathcal E^{(n)}_{t_i}),i\ge0$ ($t_0:=0$) are independent and distributed according to $\pi$.
\item For some positive constants $C,D$ we have $\p(t_{C r(n, \pi(1))}>D r(n,\pi(1)))$ vanishes polynomially. 
\end{enumerate}
Then, we have that the completion time for the same protocol on $([n],\mathcal E^{(n)}_t),t\ge0$ is $O(r(n,\pi(1)))$. 
\end{lemma}
\begin{proof}
First, we are in the presence of a sequence of independent graphs $([n],\mathcal E^{(n)}_{t_i}),i\ge0$ which are $\ER_n(\pi)$ distributed, and under the assumptions of the protocol only $O(r( n,\pi(1)))$ of them is enough to achieve completion. So, letting $N$ denote the completion time of the rumor on the original sequence of dependent graphs with $\pi$ as stationary distribution, and $\bar N$ denote the completion time of the rumor on the i.i.d. graphs with distribution $\pi$, we have
\begin{align*}
\p(N> D r(n))&\le \p(t_{\bar N}> D r(n))\\
&=\p(t_{\bar N}> D r(n),\bar N> C r(n))+\p(t_{\bar N}> D r(n),\bar N\le C r(n))\\
&\le\p(\bar N> C r(n))+\p(t_{C r(n)}> D r(n)).
\end{align*}
The first term vanishes polynomially as $n$ diverges by the conditions on the protocol and the second also, by property (2) of the $t_i,i\ge1$. 
\end{proof}

With this simple strategy, we are able to bound the completion time of information spreading protocols in settings that have not been considered previously in the literature, as far as we know. Namely, we are able to provide upper bounds for the Pull and Push-Pull protocols in the Markov dynamic and bound the Push, Pull, Push-Pull and Flood protocol for a renewal process dynamic.

\section{Application 1: The Markov case}

Let $n$ be fixed, we will omit mention to $n$ in our notation. In this case, for any $e\in E$ the process ${\bf X}(e)$ is a Markov chain, we denote such graphs edge-Markovian random graphs. Therefore $(\mathcal \varepsilon_t)_{t\ge1}$ is Markovian as well, and it has transition matrix 
\[
Q(a|b)=\prod_{e\in E}P_e(a(e)|b(e))
\]
in which $a,b\in \{0,1\}^{E}$ are binary vectors $(a(e))_{e\in E}$ and $(b(e))_{e\in E}$. 

We will investigate the homogeneous case in which  $P_e=P$ for any $e\in E$. We define $p:= P(1|0)$ and $q:= P(0|1)$. Suppose that $P$ is irreducible, then there exists a unique stationary measure $\lambda$ for $P$, and therefore $\pi$ defined by $\pi(a):= \prod_{e}\lambda(a(e))$ is the unique stationary for $Q$. 

\vspace{0.1cm}

\noindent{The \emph{maximal separation distance}}  of order $k$ of $Q$ is defined as
\[
s(k):=\max_as_a(k):=\max_a\max_b\left(1-\frac{Q^k(a,b)}{\pi(a)}\right)
\]
and observe that  
\[
s(k)=1-\left(\min_a\min_b\frac{P^k(a,b)}{\lambda(a)}\right)^{|E|}.
\]

In the binary case, 
\[
\min_a\min_b\frac{P^k(a,b)}{\lambda(a)}\ge1-|\Delta|^k\max\left\{\frac{\lambda(0)}{\lambda(1)};\frac{\lambda(0)}{\lambda(1)};1\right\}=:1-\rho|\Delta|^k,
\]
where 
\[
\Delta:=1-P(1|0)-P(0|1).
\]
So 
\[
s(k)\le 1-(1-\rho|\Delta|^k)^{|E|}.
\]
In the special case where $q= 1 -g(n)$ and $p= f(n)$, in which $f(n)$, $g(n)$ are decreasing functions with constant limit $\gamma \geq 0$,  and $|g(n) - f(n)| \leq \frac{M}{n^\alpha}$, $M,\alpha >0$.We get $\rho= h(n):=\max\left\{\frac{f(n)}{1 - g(n)}, \frac{1 - g(n)}{ f(n)}\right\}$ and $\Delta=f(n) - g(n) $, so
\[
s(k)\le 1-\left(1 - \left( |g(n) - f(n)| \right)^k h(n)\right)^{C_2^n}.
\]
For suitable constants, we will bound:

\begin{align}\label{eq:UBs}
s(k)\leq n^2\left(\frac{M}{n^\alpha}\right)^k\leq  \left(\frac{1}{n^t}\right)^{k-l}.
\end{align}

\vspace{0.1cm}

\noindent{\bf Strong uniform time.} See \cite[Section 6.2.2]{levin2017markov} for the definition of a randomized stopping time of a Markov chain. According to \cite{aldous1987strong}, a \emph{strong uniform time} for the Markov chain $\mathcal \varepsilon_t,t\ge0$ \emph{started from a given state $a_0\in E$}  is a randomized stopping time  such that $\mathcal \varepsilon_t\sim \pi$ and $\mathcal \varepsilon_t\perp T$.

Now we use Proposition 3.2 item (b)  of \cite{aldous1987strong} which states that, given any initial state $a_0$, there exists a strong uniform time $T$ for $\mathcal \varepsilon_t,t\ge0$ such that
\[
\p(T>k|\mathcal E_0=a_0)=s_{a_0}(k). 
\]
It follows that, for any $k\ge4$,
\[
\sup_{a_0}\p(T>k|\mathcal E_0=a_0)\le \left(\frac{1}{n^t}\right)^{k-l}. 
\]

\vspace{0.1cm}

{\bf Block argument: conclusion of the proof.} For any $a\in E$, denote by $T_{a}$ strong uniform time for the chain started from state $a$. Now, let us define the sequence $t_0:=0$ and $t_i=t_{i-1}+T_{\mathcal E_{t_{i-1}}}$ for any $i\ge1$. By strong Markovianess, properties (1) and (2) listed above are valid for $t_i,i\ge0$. So it remains to prove that $\p(t_{C r(n)}>D r(n))$ vanishes polynomially as $n$ diverges.

Then, by \eqref{eq:UBs}, for any $a\in E$,  $T_{a}-l\le G$ stochastically, where $G$ is Geo$(1-1/n^t)$. So we also have, stochastically, $t_i\le \sum_{j=1}^i(l+G_j)$ where $G_j,j\ge1$ are independent and identically distributed to $G$. Thus,
\begin{align*}
\p(t_{C r(n)}>D r(n))&\le  \p(\sum_{i=1}^{C r(n)}(G_i+l)> D r(n))\\
&=\mathbb{P}\left(\sum_{i = 1}^{C\text{ } r(n)} (G_i -1)> (D-C-Cl)r(n)\right)\\
&\le \exp\left(-(1-\frac{1}{s})^2 \frac{sr(n)}{2}\right).
\end{align*}
where $\frac{D}{C} - 1 - l =: s > 1$. Then, 
\[
\p(t_{C r(n)}>D r(n))\le \exp\left(-(1-\frac{1}{s})^2 \frac{sr(n)}{2}\right)
\]
which vanishes polynomially if $r(n) = \Omega(\log n)$.

As an example, consider the Push protocol spreading in a sequence of Markovian random graphs with stationary edge-probability $\pi(1) > \frac{1}{n}$. We have that the \emph{completion time} is $O(\log n)$ with high probability. Thus:
\begin{center}
\begin{equation}
\label{example}
\exp\left(-(1-\frac{1}{s})^2 \frac{s \log n}{2}\right) = n^{-(1-\frac{1}{s})^2 \frac{s}{2}}.
\end{equation}
\end{center}
For sufficiently large $s$,  Equation \eqref{example} decreases polynomially to 0 (the exponent is positive), and that is a sufficient condition to state that the Push protocol takes $O(\log n)$ steps to complete transmission in the Markov dynamic.

We have proved the following Theorem \ref{principal}

\begin{theorem}
\label{principal}
    Let $G = (G_t)_{t\geq 0} = ([n], \varepsilon_t)$ be a sequence of edge-Markovian random graphs with edge transition matrix:
    \[P = \begin{pmatrix}
    1- f(n) & f(n)\\
    1 - g(n)  & g(n)
    \end{pmatrix},\]
in which $f(n)$, $g(n)$ are decreasing functions with constant limit $\gamma \geq 0$,  and $|g(n) - f(n)| \leq \frac{M}{n^\alpha}$, $M,\alpha >0$. Let $\pi =  (\pi(0) , \pi(1))$ be the stationary distribution. Let $T_{\text{Ind}}$ be the completion time of some protocol spreading information in a sequence of independent Erd\H{o}s-Rényi graphs with parameter $\pi(1)$. Suppose that $T_{\text{Ind}}$ is $O(r(n))$ with high probability. If $r(n) = \Omega(\log n)$, then the completion time of the protocol over $G$ is $O(r(n))$ with high probability.
\end{theorem}

\subsection{Results for the Markov case}
\label{protocolos_markov}

We now state several corollaries, bounding the completion time for some information spreading protocols. The proof for these corollaries follow directly from Theorem \ref{principal} and the results mentioned in the introduction. First we consider the Push and Flood protocols.

\begin{corollary}
\label{theo:flood_push_mar}
Let $G = (G_t)_{t\geq0}$,  $G_t = ([n],\mathcal \varepsilon_t),t\ge0$, be an edge-Markovian random graph under the conditions of Theorem \ref{principal}, then
\begin{enumerate}
    \item The completion time of the Flood protocol over $G$ is $O\left(\frac{\log n}{\log( 1 + n\pi(1))}\right)$ with high probability.
    \item The completion time of the Push protocol over $G$ is $O(\frac{\log n }{\min\{1,n\pi(1)\}})$ with high probability.
\end{enumerate}

\end{corollary}

\begin{remark}
\textbf{a)} Using a coupling argument and part (1) of Corollary \ref{theo:flood_push_mar} , we can consider a new class of edge-Markovian random graphs that do not fall under the conditions of Theorem \ref{principal} for the Flood protocol. Let $G' = (G'_t)_{t\geq 0} = ([n], \varepsilon'_t)$ be a sequence of edge-Markovian random graphs with transition matrix 
\[P' = \begin{pmatrix}
    1- \frac{a}{n^k} & \frac{a}{n^k}\\
    1  & 0
    \end{pmatrix},\] 
    $a>0$ and any $k$. The completion time of the Flood protocol over $G'$ is $O\left(\frac{\log n}{\log( 1 + \frac{1}{n^{k-1}})}\right)$ with high probability. Let $G'' = (G''_t)_{t\geq 0} = ([n], \varepsilon''_t)$ be a sequence of edge-Markovian random graphs with transition matrix \[P'' = \begin{pmatrix}
    1- \frac{a}{n^k} & \frac{a}{n^k}\\
    1-\alpha   & \alpha
    \end{pmatrix},\] $\alpha >0$. We can couple $\varepsilon'_t$ and $\varepsilon''_t$ such that $\varepsilon'_t \subseteq \varepsilon''_t$, $\forall t\ge 1$. Let $I'_t$ and $I''_t$ be the set of informed nodes at time $t$ in $G'$ and $G''$, respectively. Since the Flood protocol informs every neighbor of informed nodes each step, we have that $I'_t \subseteq I''_t$. Therefore the completion time on $G''$ is at most the completion time on $G'$. We conclude that the completion time of the flood protocol over $G''$ is $O\left(\frac{\log n}{\log( 1 + \frac{1}{n^{k-1}})}\right)$.\\

    \textbf{b)} We also note that the Push protocol takes at least $\log_2 n$ steps to complete transmission (the number of informed nodes at most doubles, each time step). So when $\min\{n\pi(1), 1\} = 1$, the completion time is $\Theta(\log n)$.
\end{remark}

The following Corollary \ref{caso_especial} highlights a special case of part (2) of Corollary \ref{theo:flood_push_mar} over a sequence of edge-Markovian graphs where $p = \Theta(\frac{1}{n^k})$, $k > 1$. As far as we know, this case was not considered previously in the literature.

\begin{corollary}
\label{caso_especial}
     Let $([n],\mathcal \varepsilon_t)$, $t\ge0$, be a sequence of edge-Markovian random graphs with $f(n) = \frac{a}{n^k}$ and $g(n) = 0$, in which $a>0$ and $k>1$ are constants. The completion time of the Push protocol over $G$ is $O(n^{k-1}\log n)$
\end{corollary}

We now consider the Pull and Push-Pull protocols. As far as we know, this is the first time bounds for the completion time are given for these protocols on edge-Markovian graphs. The conditions for these protocols are more restrictive, only allowing for transition matrices with $\pi(1) = \frac{a}{n}$.

\begin{corollary}\label{theo:pull_push_pull_mar}
Let $G = (G_t)_{t\geq0}$,  $G_t = ([n],\mathcal \varepsilon_t),t\ge0$, be an edge-Markovian random graph under the conditions of Theorem \ref{principal} with $\pi(1) = \frac{a}{n}$, $a >0$, then
\begin{enumerate}
    \item The completion time of the Pull protocol over $G$ is $O(\log n)$ with high probability.
    \item The completion time of the Push-Pull protocol over $G$ is $O(\log n)$ with high probability.
\end{enumerate}
\end{corollary}

\section{Application 2: beyond Markov}
\subsection{Stationary Times Through Coupling From the Past}
\label{STCFTP}
In this Section we introduce a method to obtain strong stationary times for any process that can be perfectly simulated: for our purposes, we use that to mean we can construct an almost surely finite \emph{coupling from the past} (CFTP) algorithm. CFTP is an algorithm to sample exactly from the stationary distribution of a stochastic chain. It was introduced in the 1990s in the context of Markov chains (see the seminal paper \cite{propp/wilson/1996}) and later extended to a class of non-Markovian chains (see \cite{comets/fernandez/ferrari/2002}, \cite{gallo/2009}, \cite{gallo/garcia/2013}). In what follows we give a basic definition of the CFTP algorithm and show how it can be used, coupled with our strategy Lemma \ref{strategy}, to bound the completion time of randomized rumor spreading protocols. As far as we know, there is no existing theory in the literature for stationary times in non-Markovian chains.

Let $A$ be a finite alphabet. Let $A^{-\mathbb{N}}$ be the set of infinite strings of past symbols. Let $(X_t)_{t \in \mathbb{Z}}$ be a discrete time stochastic process in $A$. Let \[
P: A \times A^{-\mathbb{N}} \to [0,1]\]
\[
(a , \hat a) \to P(a | \hat a)
\]
be the transition probability kernel. A stationary sequence of random variables $(X_i)_{i \in \mathbb{Z}}$ is said to be compatible with P if
\[
\text{Prob}(X_0 = a_0|X_{-1} = a_{-1},X_{-2} = a_{-2},\dots) = P(a|a_{-1},a_{-2}, \dots).
\]
Let $(U_t)_{t\in\mathbb{Z}}$ be a sequence of i.i.d. random variables uniformly distributed in $[0,1]$. We construct an update function $f: A^{-\mathbb{N}} \times [0,1] \to A$ such that:
\[
\mathbb{P}(f(U_{t+1}, \hat a) = a) = P(a | \hat a),
\]
for any $a \in A$, $\hat a \in A^{-\mathbb{N}}$. Let
\[
X_{t+1} = f(U_{t+1}, X_t, X_{t-1},\dots).
\]
Consequently, $(X_{t})_{t \in \mathbb{Z}}$ is compatible with $P$.
Let $k < l$ be integers, through iterations of $f$ we obtain a function $f_{k,l} : A^{-\mathbb{N}} \times [0,1]^k \to A$  allowing us to see the $k$-th symbol in a sequence with a fixed past $\hat a \in  A^{-\mathbb{N}}$ using $U_k, U_{k+1}, \dots U_l$. 
We define the coalescing time \[
\theta_{k}:= \max\{i\leq k : f_{i,k}(\hat a, U_i, \dots, U_k)\text{ does not depend on $\hat a$}\}.\]

If we can prove that the coalescing time is almost surely finite, coupling from the past is the following algorithm, used to sample from the stationary distribution of $(X_t)$.

\begin{itemize}
    \item Compute $\theta_{0}$.
    \item Compute $X_0 = f_{\theta_{0},0}(\hat a, U_{\theta_0},\dots,U_0)$ (does not depend on $\hat a$).
    \item $X_0$ is a sample of the unique stationary chain compatible with $P$.
\end{itemize}

Suppose that for $(X_t)_{t \in \mathbb{Z}}$, we have that $\theta_0$ is almost surely finite. We can now construct a sequence of strong stationary times for $(X_t)_{t \in \mathbb{Z}}$. Let $\tau_1 : = \theta_0$. Suppose we repeat the coupling from the past algorithm starting at $\tau_1 - 1$. There exists a second finite coalescing time $\tau_2 := \theta_{\tau_1 - 1}$ and the algorithm yields that $X_{\tau_1 -1}$ is distributed according to the stationary distribution of the process and is independent of $X_0$.

We repeat this $N-2$ times and conclude that $t_1 = 0$, $t_2 = \tau_1 -1 $, $t_3 = \tau_2 - 1$,..., $t_N = \tau_{N-1} - 1$ is a sequence of strong stationary times for $(X_n)$: i.i.d. random times such that $X_{t_i}$ follows the stationary distribution. Combined with Lemma \ref{strategy}, this yields the following Theorem \ref{cftp_rumor}.

\begin{theorem}
\label{cftp_rumor}
     Let $([n], \varepsilon_t),{t\in \mathbb{Z}},$ be a sequence of non-independent random graphs. Consider an information spreading protocol. Suppose that its completion time in the i.i.d. case is $O(r(n))$ with high probability. If there exists an almost surely finite coupling for the past algorithm for $(\varepsilon_t)_{t \geq 0}$ such that, for sufficiently large constants $C,D$, $\lambda >0$:
     \[
     \mathbb{P}(-t_{C r(n)} > D r(n)) = \frac{1}{n^\lambda},
     \]
     the completion time of the protocol over $(G_t)$ is $O(r(n))$ with high probability.
\end{theorem}

\subsection{Example 2: The renewal case.}

In this case, for every $e \in E $, the process  ${\bf X}(e)$ is a binary stationary renewal process, with renewal state $0$ (absent edge). For each edge $e \in E$ we sample i.i.d. random times (non-negative random variables with finite mean) $(Z_i^e)_{i\geq 1}$. Their distribution may depend on the number of nodes $n$, we will omit $n$ in our notation. Let $\mu := E(Z_1^e)$. We then define $Z_0^e$ as a random variable with distribution
\[
\text{Prob}(Z_0^e = t) := \frac{\text{Prob}(Z_1^e > t)}{\mu},
\]
for any $t \in \mathbb{N}$, and let  $X_{\sum_{i = 0}^k Z_i^e}(e) = 0$, for $k\geq 0$, and $X_t(e) = 1$ for every other instant $t$. This construction guarantees that the renewal process is stationary, with marginal stationary distribution $\pi = (\pi(0), \pi(1))$, where $\pi(0) = \frac{1}{\mu}$. We name such graphs edge-renewal random graphs. Since an edge is independent of every other edge, we have that the joint process for all edges $(\varepsilon_t)_{t \geq 1}$ is a renewal process as well, with renewal state $\{0\}^{|E|}$ (every edge is absent). Under sufficient conditions, we prove that there exists a random sequence of integers $t_1,t_2,\dots,t_N$ such that: $([n], \varepsilon_{t_i})$ are independent and stationary.

Suppose that for every $e \in E$, 
\begin{center}
    \begin{equation}
    \label{minori}
        \inf_i \text{Prob}\left(\sum_{k=1}^j Z_k^e =  i+1 | \sum_{k = 1}^j Z_k^e \geq i\right)= 1 - \alpha > 0,
    \end{equation}
\end{center}

Let $((U_{-i}^e)_{i\geq 0})_{e\in E}$ be a sequence of i.i.d. random variables with uniform distribution in $[0,1]$. \cite{gallo/2009} shows that the following coalescing time $\theta_0 = \inf\{i \geq 0: U_{-i}^e \leq 1 - \alpha\text{ } \forall \text{ } e \in E\}$ is almost surely finite. Therefore we can perfectly simulate from the process. We define $t_1,t_2,\dots,t_N$  as in Section \ref{STCFTP}. Since $\theta_0$ is geometric, we have that, for any $i \in \{2,3,\dots,N\}$,
    \begin{center}
    \begin{equation}
    \label{geometrica}
    \mathbb{P}(|t_i - t_{i-1}| > k) = \left(1-(1-\alpha)^{n \choose 2}\right)^k,
    \end{equation}
    \end{center}
which will be useful in what follows.

 We can now state a Corollary of Theorem \ref{cftp_rumor}.

\begin{corollary}
    
\label{rumor_renewal}
    Let $([n], \varepsilon_t), t \in \mathbb{N}$, be an edge-renewal random graph satisfying \eqref{minori} with $\alpha = \alpha_n := \frac{g(n)}{n^\lambda}$, in which $\lambda > 0$ and $g(n)$ is a function with constant limit $\gamma \geq 0$ such that $\frac{g(n)}{n^\lambda}$ decreases polynomially to $0$ as $n$ diverges and $\alpha \in (0,1)$. Let $T_{\text{Ind}}$ be the completion time of some protocol spreading information in a sequence of independent Erd\H{o}s-Rényi graphs with parameter $\pi(1)$. Suppose that $T_{\text{Ind}}$ is $O(r(n))$ with high probability. If $r(n) = \Omega(\log n)$, then the completion time of the protocol over $G$ is $O(r(n))$ with high probability.
\end{corollary}

\begin{proof}
Since the edge process can be perfectly simulated there exists a sequence of strong stationary times for this process, therefore, for the strategy in Lemma \ref{strategy} to hold it remains to prove that $\p(-t_{C r(n)}>D r(n))$ vanishes polynomially as $n$ diverges, for sufficiently large constants $C$ and $D$. Let $(G_i)_{i\geq 1}$ be a sequence of i.i.d. random variables geometrically distributed with parameter $1 - \frac{g(n)}{n^\lambda}$. Let $s = \frac{D}{C} - 1$. From \eqref{geometrica}, it follows that
\[ 
\p(-t_{Cr(n)} > Dr(n)) = \p\left(\sum_{i=1}^{Cr(n)} G_i > Dr(n)\right)
\]
\[
= \p\left(\sum_{i=1}^{Cr(n)} (G_i-1) > (D-C)r(n)\right)
\]
\[
\leq \exp\left(-\left(1 - \frac{1}{s}\right)^2\frac{sr(n)}{2}\right),
\]
which vanishes polynomially if $r(n) =\Omega(\log n)$.
\end{proof}

We note that the Corollaries in Section \ref{protocolos_markov} can also be stated for the renewal process dynamic.

For sufficiently large $n$, Corollary \ref{rumor_renewal} holds for the edge-renewal random graph $G$ with \[\text{Prob}\left(Z_1^e =  i+1 \bigg{|}  Z_1^e \geq i\right) = 1 - \frac{1}{n^\lambda}\frac{i+2}{i+1},\]
in which $\lambda \geq 1$. Note that in this case $\alpha = \alpha_n := \frac{1}{n^\lambda}$. The mean renewal time for an edge $e$ is: \[
\mu = \sum_{k=1}^\infty \text{Prob}(T_1^e \geq k) =1 +  \sum_{k=2}^\infty \prod_{i =2 }^k \frac{1}{n^\lambda}\frac{i+2}{i+1} = 1 + \sum_{k=2}^\infty k\frac{1}{n^{(k-1)\lambda}} = \]
\[
1 + \sum_{k=1}^\infty (k+1)\frac{1}{n^{k\lambda}} = 1 + \frac{1}{n^\lambda} \sum_{k=1}^\infty k\frac{1}{n^{(k-1)\lambda}} + \sum_{k=0}^\infty \frac{1}{n^{k\lambda}} - 1 = \frac{n^\lambda}{(n^\lambda -1)^2 } + \frac{n^\lambda}{n^\lambda - 1} = \Theta\left(1\right).
\]
Therefore the stationary edge-probability is:
\[
\pi(1) = 1 -\frac{1}{\mu} = 1 - \frac{(n^\lambda - 1)^2}{ n^\lambda + (n^\lambda - 1)n^\lambda}=\frac{2n^\lambda -1 }{n^\lambda} =\Theta\left(\frac{1}{n^\lambda}\right).
\]
Then, for example, the completion time of the Push protocol over $G$ is $O\left(n^{\lambda - 1}\log n\right)$.

\vspace{1cm}

\bibliographystyle{jtbnew}
\bibliography{sandro_biblio}

\end{document}